\documentclass{scrartcl}
\usepackage{enumitem}
\newlist{dsitemize}{itemize}{1}
\setlist[dsitemize,1]{label=--,leftmargin=6mm}
\usepackage{amsmath,amssymb}
\usepackage{mathtools}
\usepackage{cleveref}
\usepackage{booktabs}
\usepackage{todonotes} %
\usepackage[caption = false]{subfig}
\usepackage{tikz}
\usetikzlibrary{arrows,decorations.pathmorphing}
\tikzset{
	vertex/.style={circle,draw,inner sep=01pt,minimum size=1.3em},
	vertex2/.style={circle,draw,inner sep=01pt,minimum size=1.5em},	
	vertex2g/.style={color=gray,circle,draw,inner sep=01pt,minimum size=1.5em},		
	vertexb/.style={draw,inner sep=01pt,minimum size=1.3em},	
	dedge/.style={->,> = latex', font=\footnotesize},
}
\usepackage{algorithm}%
\usepackage[noend]{algpseudocode}%
\newcommand{\tg}{{G}}
\newcommand{\tge}{{E}}
\newcommand{\probaenum}{\textsc{MceaEnum}}
\newcommand{\probas}{\textsc{Mcea}}
\newcommand{\probaenumsp}{$($$c$$>$$0$$)$-\textsc{MceaEnum}}
\newcommand{\probacount}{\textsc{\#Mcea}}
\newcommand{\tpeac}{\textsc{\#Eap}}
\newcommand{\tpfc}{\textsc{\#Fp}}
\newcommand{\probbenum}{\textsc{McfEnum}}
\newcommand{\probbs}{\textsc{Mcf}}
\newcommand{\probbenumsp}{$($$c$$>$$0$$)$-\textsc{McfEnum}}
\newcommand{\probbcount}{\textsc{\#Mcf}}

\newcommand{\enump}{$\mathcal{E}$}

\newcommand{\delayp}{\textbf{DelayP}}
\newcommand{\psdelayp}{\textbf{PSDelayP}}
\newcommand{\spcomplete}{\textbf{\#P}-complete}
\newcommand{\spcompleteness}{\textbf{\#P}-completeness}
\newcommand{\spc}{\textbf{\#P}}
\newcommand{\prefixp}{prefix-path}
\newcommand{\suffixp}{suffix-path}
\usepackage[misc]{ifsym}
\usepackage{amsthm}
\newtheorem{definition}{Definition}
\newtheorem{corollary}{Corollary}
\newtheorem{observation}{Observation}
\newtheorem{lemma}{Lemma}
\newtheorem{theorem}{Theorem}
\usepackage{graphicx}
\usepackage{authblk}
\begin{document}
\title{On the Enumeration and Counting of Bicriteria Temporal Paths\footnote{A preliminary conference version of this work appeared in~\cite{mutzel2019enumeration}.}}
\author{Petra Mutzel}
\author{Lutz Oettershagen}
\affil{Institute for Computer Science, University of Bonn}
\affil[ ]{\texttt {\{petra.mutzel,lutz.oettershagen\}@cs.uni-bonn.de}}
\renewcommand*{\thefootnote}{\fnsymbol{footnote}}
\maketitle
\setcounter{footnote}{0}
\renewcommand*{\thefootnote}{\arabic{footnote}}
\begin{abstract}
	We discuss the complexity of path enumeration and counting in weigh\-ted temporal graphs.
	In a weighted temporal graph, each edge has an availability time, a traversal time and some real cost. 
	We introduce two bicriteria temporal min-cost path problems in which we are interested in the set of all efficient paths with low cost and short duration or early arrival time, respectively.
	However, the number of efficient paths can be exponential in the size of the input.
    For the case of strictly positive edge costs, we are able to provide algorithms that enumerate the set of efficient paths with polynomial time delay and polynomial space.
    If we are only interested in the set of Pareto-optimal solutions and not in the paths themselves, then these can be determined in polynomial time if all edge costs are nonnegative. 
	In addition, for each Pareto-optimal solution, we are able to find an efficient path in polynomial time.
	On the negative side, we prove that counting the number of efficient paths is \spcomplete, even in the non-weighted single criterion case.
\end{abstract}
\section{Introduction}
A weighted temporal graph $G=(V,E)$ consists of a set of vertices and a set of temporal edges.
Each temporal edge $e\in E$ is associated with an edge cost and is only available (for departure) at a specific integral point in time. 
Traversing an edge takes a specified amount of traversal time. 
We can imagine a temporal graph as a sequence of $T\in \mathbb{N}$ graphs $G_1,G_2,\ldots,G_T$ sharing the common set of vertices $V$, and each graph $G_i$ has its own set of edges $E_i$.
Therefore, $G$ can change its structure over a finite sequence of integral time steps. 

Given a directed weighted temporal graph $G=(V,E)$, a source $s\in V$ and a target $z$\footnote{We use $t$ to denote \emph{time steps} and $z$ for the target vertex.}$\in V$, our goal is to find \emph{earliest arrival} or \emph{fastest} $(s,z)$-paths with \emph{minimal costs}.
A motivation can be found in typical queries in (public) transportation networks.
Here, each vertex represents a bus stop, metro stop or a transfer point and each edge a connection between two such points.
In this model, the availability time of an edge is the departure time of the bus or metro, 
the traversal time is the time a vehicle takes between the two transfer points, 
and the edge cost represents the ticket price. 
Two natural occurring problems are the following: (1) Minimize the costs and the arrival times, or (2) minimize the costs and the total travel time. %

In general, there is no path that minimizes both objectives simultaneously, and therefore we are interested in the set of all \emph{efficient} paths. 
A path is called {efficient} if there is no other path that is strictly better in one of the criteria and at least as good concerning both criteria. In other words, a path is efficient iff its cost vector is Pareto-optimal.

\probbenum{} and \probaenum{} denote the enumeration problems, in which the task is to enumerate exactly all efficient paths w.r.t.\ cost and duration or cost and arrival time, respectively.
We show, that there can be an exponential number of efficient paths. So one cannot expect to find polynomial time algorithms for these problems.
However, 
Johnson, Yannakakis, and Papadimitriou~\cite{JOHNSON1988119} have defined complexity classes for enumeration problems, where the time complexity is expressed not only in terms of the input size but also in the output size. We use the output complexity model to analyze the proposed enumeration problems and show that they belong to the class of polynomial time delay with polynomial space (\psdelayp{}).
If we are only interested in the sets of Pareto-optimal solutions and not the paths themselves, we show that
these problems can be solved in polynomial time for nonnegative edge weights. In these cases, we can also provide an associated path with each solution. 
On the negative side, we show that counting the number of efficient paths is not easier than their enumeration, i.e., determining the number of efficient paths is hard.
\\\\\textbf{Contribution}~--
In this paper we show the following: %
\begin{enumerate}%
	\item \probbenum{} and \probaenum{} are in \psdelayp{} for weighted temporal graphs with strictly positive edge costs. %
	\item 	In case of nonnegative edge costs, finding the Pareto-optimal set of cost vectors is possible in polynomial time (thus the set of Pareto-optimal solutions is polynomially bounded in the size of the input), and
	for each Pareto-optimal solution we can find an efficient path in polynomial time. %
	\item The decision versions that ask if a path is efficient, are in \textbf{P}. 
	Deciding if there exists an efficient path with given cost and duration or arrival time is possible in polynomial time. %
	\item The counting versions are \spcomplete, even for the single criterion unweighted case.
\end{enumerate}
Our paper has the following structure. In the remainder of this section we discuss the related work. In \Cref{section:preliminaries} we provide all necessary preliminaries. Next, in \Cref{section:structuralresults} structural results are presented. These are the foundation of the algorithms for \probbenum{} and \probaenum{} in the following \Cref{section:algprobb,section:algproba}, respectively. 
\Cref{section:counting} discusses the complexity of counting efficient paths.
Finally, in \Cref{section:conclusion}  conclusions are drawn and we state further open problems.
\\\\\textbf{Related Work}~--
Temporal graphs and related problems are discussed in several recent works. 
For a general overview there are a number of nice introductions, e.g., 
\cite{casteigts2012time,holme2015modern,holme2012temporal,kostakos2009temporal}.
An early work on graphs with time-depended travel times is from Cooke and Halsey~\cite{cooke1966shortest}, in which they provide a generalized formulation for the Bellman iteration scheme~\cite{bellman1958routing} in order to support varying travel times.
Kempe et al.~\cite{kempe2002connectivity} analyze connectivity problems with respect to temporal paths.
Xuan et al.~\cite{Xuan2003} discuss communication in dynamic and unstable networks. They introduce algorithms for finding a fastest path, an earliest arrival path and a path that uses the least number of edges. 
Wu et al.~\cite{wu2014path} further discuss the fastest, shortest and earliest arrival path problems in temporal graphs and introduce the latest-departure path problem.
Their algorithm for calculating an earliest arrival $(s,v)$-path has time complexity $\mathcal{O}(n+m)$ and  space complexity $\mathcal{O}(n)$, where $n$ denotes the number of vertices and $m$ the number of edges  in the given temporal graph. 
Furthermore, they present an algorithm that finds a fastest $(s,v)$-path in $\mathcal{O}(n+m\log c)$ time and $\mathcal{O}(\min\{n\cdot\mathcal{S},n+m\})$ space. 
Here, $\mathcal{S}$ is the number of distinct availability times of edges leaving vertex $s$, and $c$ the minimum of $\mathcal{S}$ and the maximal 
in-degree over all vertices of $G$. 
The algorithm uses a label setting approach to find a fastest $(s,v)$-path for each $v\in V$.
Both,~\cite{Xuan2003} and~\cite{wu2014path}, consider only non-weighted temporal graphs.

Hansen~\cite{HansenBSP} introduces bicriteria path problems in static graphs, %
and provides an example for a family of graphs for which the number of efficient paths grows exponentially with the number of vertices.
Meggido shows that deciding if there is an path that respects an upper bound on both objective functions is \textbf{NP}-complete (Meggido 1977, private communication with Garey and Johnson~\cite{gareyjohnson}).
Martins~\cite{martins1984multicriteria} presents a label setting algorithm based on the well known Dijkstra algorithm for the bicriteria shortest path problem, that finds the set of all efficient $(s,v)$-paths for all $v\in V$.
Ehrgott and Gandibleux~\cite{ehrgott2000survey} provide a good overview of the work on bi- and multicriteria shortest path problems.
Hamacher et al.~\cite{hamacher2006algorithms} propose an algorithm for the bicriteria time-dependent shortest path problem
in networks in which edges have time dependent costs and traversal times.
The traversal time of an edge is given as a function of the time upon entering the edge. 
Moreover, each edge has a two-dimensional time dependent cost vector. Waiting at a vertex may be penalized by additional bicriteria time dependent costs. 
They propose a label setting algorithm that starts from the target vertex and finds the set of all efficient paths to each possible start vertex.
Disser et al.~\cite{disser2008multi} discuss a practical work on the multi-criteria shortest path problem in time-dependent train networks.
To this end, they introduce foot-paths and special transfer rules to model realistic train timetables. 

We are not aware of any work discussing the enumeration or counting of the set of efficient paths in weighted temporal graphs.
\section{Preliminaries}\label{section:preliminaries}
A \emph{static, directed} \emph{graph} $G=(V, E)$ consists of a finite set of vertices $V$ and finite set $E\subseteq \{(u,v)\in V\times V \mid u\neq v\}$ of directed edges.
A weighted temporal graph $\tg=(V,\tge)$ consists of a set $V$ of $n\in \mathbb{N}$ vertices and a set $E$ of $m\in\mathbb{N}$ weighted and directed temporal edges. 
A weighted and directed temporal edge $e=(u,v,t,\lambda,c)\in\tge$ consists of the starting vertex $u\in V$, the end vertex $v\in V$ (with $u\neq v$), 
 availability time $t\in\mathbb{N}$, 
traversal time $\lambda\in\mathbb{N}$ and cost $c\in\mathbb{R}_{\geq 0}$. %
Each edge $e=(u,v,t,\lambda,c)\in\tge$ is only available for entering at its availability time $t$ and traversing $e$ takes $\lambda$ time.
For $T\coloneqq\max\{t\mid (u,v,t,\lambda,c)\in\tge\}$, we can view $\tg$ as a finite sequence $G_1,G_2,\ldots,G_{T}$ of static graphs over the common set of vertices $V$. Each $G_i$ has its own set of edges $E_i\coloneqq\{(u,v,\lambda,c)\mid e=(u,v,i,\lambda,c)\in\tge\}$. \Cref{fig:example_tg} shows an example.
We denote the set of incoming (outgoing) temporal edges of a vertex $v\in V$ by $\delta^-(v)$ ($\delta^+(v)$). 
Note that in general for temporal graphs the number of edges is not bounded by a function in the number of vertices, i.e. we may have arbitrarily more temporal edges than vertices.
In the following, we use a \emph{stream representation} of temporal graphs. 
A temporal graph is given as a sequence of its $m$ edges, which is ordered by the availability time of the edges in increasing order with ties being broken arbitrarily.
\subsection{Temporal Path Problems}\label{subsection:tgpath}
A \emph{temporal} $(u,v)$-\emph{walk} $P_{u,v}$ is a sequence $(e_1,\ldots,e_i=(v_i,v_{i+1},t_i,\lambda_i,c_i),\ldots, e_k)$ of edges with $e_i\in \tge$ for $1\leq i\leq k$, and with $v_1=u$, $v_{k+1}=v$ and $t_i+\lambda_i\leq t_{i+1}$ for $1\leq i <k$.
If a temporal $(u,v)$-walk visits each $w\in V$ at most once, it is \emph{simple} and we call it $(u,v)$-\emph{path}.
We denote by $s(P_{u,v}) \coloneqq t_1$ the \emph{starting time}, and by $a(P_{u,v})\coloneqq t_k+\lambda_k$ the \emph{arrival time} of $P_{u,v}$. 
Furthermore, we define the \emph{duration} as $d(P_{u,v})\coloneqq a(P_{u,v})-s(P_{u,v})$. 
A path $P_{u,v}$ is \emph{faster} than a path $Q_{u,v}$ if $d(P_{u,v})<d(Q_{u,v})$.
The \emph{cost} of a path $P=(e_1,\ldots, e_k)$ is the sum of the edge cost, i.e. $c(P)\coloneqq \sum_{i=1}^{k}c_i$. 
Finally, for a path $P=(e_1,\ldots,e_i,\ldots,e_k)$, we call $(e_1,\ldots,e_i)$ \emph{\prefixp{}} and $(e_i,\ldots,e_k)$ \emph{\suffixp{}} of $P$. %
In \Cref{fig:example_tg} the $(s,z)$-path $((s,b,1,1,2),(b,z,2,1,1))$ has arrival time $3$, duration $2$ and cost $3$. The $(s,z)$-path $((s,z,3,1,3))$ also has cost $3$, but it has a later arrival time of $4$ and is faster with a duration of only $1$.
\begin{figure}
\centering
	\begin{tikzpicture}
	\node at (-1.5,0.75) {$G:$};
	\node at (1.5,.75) {$G_1:$};
	\node at (4.8,.75) {$G_2:$};
	\node at (8,.75) {$G_3:$};

	\node[vertex] (1) at (-1.5,0) {s};
	\node[vertex] (2) at (0,0) {a};
	\node[vertex] (3) at (0,-1.5) {b};
	\node[vertex] (4) at (-1.5,-1.5) {z};
	\path[->] 
	(1)  edge[dedge]  node[midway,above,sloped] {\scriptsize$(1,1,1)$}   (2)
	(1)  edge[dedge]  node[midway,below,sloped] {\scriptsize$(3,1,3)$}   (4)
	(2)  edge[dedge]  node[midway,above,sloped] {\scriptsize$(2,1,2)$}   (3)
	(1)  edge[dedge]  node[midway,above,sloped] {\scriptsize$(1,1,2)$}   (3)
	(3)  edge[dedge]  node[midway,below,sloped] {\scriptsize$(2,1,1)$}   (4);	
	
	\node[vertex] (1) at (1.5,0) {s};
	\node[vertex] (2) at (3,0) {a};
	\node[vertex] (3) at (3,-1.5) {b};
	\node[vertex] (4) at (1.5,-1.5) {z};
	\path[->] 
	(1)  edge[dedge]  node[midway,above,sloped] {\scriptsize$(1,1)$}   (2)
	(1)  edge[dedge]  node[midway,above,sloped] {\scriptsize$(1,2)$}   (3);	
	
	\node[vertex] (1) at (4.8,0) {s};
	\node[vertex] (2) at (6.3,0) {a};
	\node[vertex] (3) at (6.3,-1.5) {b};
	\node[vertex] (4) at (4.8,-1.5) {z};
	\path[->] 
	(3)  edge[dedge]  node[midway,below,sloped] {\scriptsize$(1,1)$}   (4)	
	(2)  edge[dedge]  node[midway,above,sloped] {\scriptsize$(1,2)$}   (3);	
	\node[vertex] (1) at (8,0) {s};
	\node[vertex] (2) at (9.5,0) {a};
	\node[vertex] (3) at (9.5,-1.5) {b};
	\node[vertex] (4) at (8,-1.5) {z};
	\path[->] 
	(1)  edge[dedge]  node[midway,below,sloped] {\scriptsize$(1,3)$}   (4)	;	
	\end{tikzpicture}
	\caption{Example for a weighted temporal graph $\tg$. Each edge label $(t,\lambda,c)$ describes the time $t$ when the edge is available, its traversal time $\lambda$ and its cost~$c$. For each time step $t\in\{1,2,3\}$, layer $G_t$ is shown.}
	\label{fig:example_tg}
\end{figure}
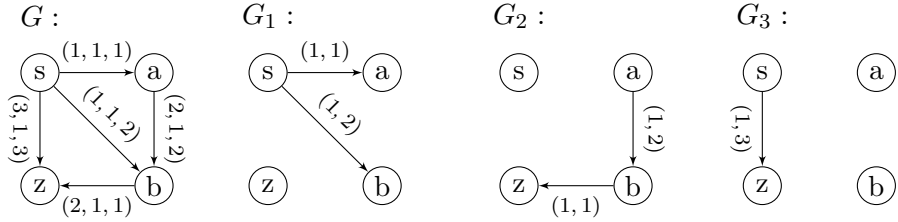

For the discussion of bicriteria path problems, we use the following definitions. 
Let $\mathcal{X}$ be the set of all feasible $(s,z)$-paths, and let $f(P)$ be the temporal value of $P$, i.e. either arrival time $f(P)\coloneqq a(P)$ or duration $f(P)\coloneqq d(P)$.
We call a path $P\in\mathcal{X}$ \emph{efficient} if there is no other path $Q\in\mathcal{X}$ with $c(Q) < c(P)$ and $f(Q) \leq f(P)$ or $c(Q) \leq c(P)$ and $f(Q) < f(P)$.
We map each $P\in\mathcal{X}$ to a vector $(f(P),c(P))$ in the two-dimensional objective space which we denote by $\mathcal{Y}$.
Complementary to efficiency in the decision space, we have the concept of \emph{domination} in the objective space. 
We say $(f(P), c(P))\in\mathcal{Y}$ \emph{dominates} $(f(Q), c(Q))\in\mathcal{Y}$ if either $c(P) < c(Q)$ and $f(P) \leq f(Q)$ or $c(P) \leq c(Q)$ and $f(P) < f(Q)$. 
We call $(f(P), c(P))$ \emph{nondominated} if and only if $P$ is efficient.
We define the bicriteria enumeration problems \probbenum{} and \probaenum{} as follows.
\\\\\textsc{Min-Cost Fastest Paths Enumeration Problem} (\probbenum)
\\\textbf{Given:} A weighted temporal graph $\tg=(V,\tge)$ and $s,z\in V$.
\\\textbf{Task:} Enumerate all and only $(s,z)$-paths that are efficient w.r.t. duration and costs.
\\\\\textsc{Min-Cost Earliest Arrival Paths Enumeration Problem} (\probaenum)
\\\textbf{Given:} A weighted temporal graph $\tg=(V,\tge)$ and $s,z\in V$.
\\\textbf{Task:} Enumerate all and only $(s,z)$-paths that are efficient w.r.t. arrival time and costs.
\\

We denote by \probbs{} and \probas{} the optimization versions, in which the task is to find a single efficient $(s,z)$-path. 
The counting versions \probacount{} and \probbcount{} ask for the number of solutions. 
\subsection{Complexity Classes for Enumeration and Counting Problems}\label{subsection:outputcomplex}
Bi- and multicriteria optimization problems are often not easily comparable using the traditional notion of worst-case complexity, due to their potentially exponential number of efficient solutions. 
We use the \emph{output complexity} model as proposed by Johnson, Yannakakis, and Papadimitriou~\cite{JOHNSON1988119}.
Here, the time complexity is stated as a function in the size of the input and the output.
\begin{definition}
	Let \enump{} be an enumeration problem. Then \enump{} is in
	\begin{enumerate}
		\item \emph{\delayp{} (Polynomial Time Delay)} if the time delay until the output of the first and between the output of any two consecutive solutions is bounded by a polynomial in the input size, and
		\item \emph{\psdelayp{} (Polynomial Time Delay with Polynomial Space)} if \enump{} is in \delayp{} and the used space is also bounded by a polynomial in the input size.
	\end{enumerate}
\end{definition}
In \Cref{section:algprobb,section:algproba} we show that \probbenum{} and \probaenum{} are both in \psdelayp{} if the input graph has strictly positive edge costs. We provide algorithms that enumerate the set of all efficient paths in polynomial time delay and use space bounded by a polynomial in the input size.

By enumerating the set of all efficient paths, we are also able to count them, i.e., enumeration is at least as hard as counting. Naturally, the question arises, if we can obtain the number of efficient paths more easily. Valiant~\cite{valiant1979complexity} introduced the class \spc{}. 
We use the following definitions from Arora and Barak~\cite{arora2009computational}.
\begin{definition}
	A function $f:\{0,1\}^*\rightarrow \mathbb{N}$ is in \emph{\spc{}} if there exists a polynomial $p:\mathbb{N}\rightarrow\mathbb{N}$ and a deterministic polynomial time Turing machine $M$ such that for every $x\in\{0,1\}^*$ $f(x) = |\{y\in \{0,1\}^{p(x)}: M(x,y)=1 \}|\text{,}$
	where $M(x,y)=1$ iff $M$ accepts $x$ with certificate $y$. %
\end{definition}
In other words, the class \spc{} consists of all functions $f$ such that $f(x)$ is equal to the number of computation paths from the initial configuration to an accepting configuration of a polynomial time non-deterministic Turing machine $M$ with input $x$.

In the following, we use the notion of oracle Turing machine where the oracles are not limited to answer only boolean, but arbitrary, polynomially bounded functions.
Let $f:\{0,1\}^*\rightarrow\{0,1\}^*$ be such a function. We define $\textbf{FP}^f$ as the class of functions that can be computed by a polynomial time Turing machine with constant-time access to $f$.
\begin{definition}
	A function $f$ is \emph{\spcomplete} iff it is in \emph{\spc} and $\emph{\spc}\subseteq\textbf{FP}^f$.
\end{definition}
A polynomial time algorithm for a \spcomplete{} problem would imply $\textbf{P}=\textbf{NP}$.
To show \spcompleteness{} for a counting problem $\mathcal{C}$, we verify that $\mathcal{C}\in \spc$ and reduce a \spcomplete{} problem $\mathcal{C}'$ to $\mathcal{C}$,
using a polynomial Turing reduction.
Valiant~\cite{valiant1979complexity} showed that calculating the permanent of a $(0,1)$-matrix is \spcomplete. This is especially interesting because the permanent of a $(0,1)$-matrix corresponds to the number of perfect matchings in a bipartite graph. Deciding if a bipartite graph has a perfect matching is possible in polynomial time. 
Moreover, Valiant~\cite{valiant1979complexityenum} proved \spcompleteness{} for \textsc{ST-Path}.
The input of \textsc{ST-Path} is a static, directed graph $\tg=(V,E)$ and $s,t\in V$, the output is the number of $(s,t)$-paths. %
In \Cref{section:counting}, we give a reduction from \textsc{ST-Path} to the non-weighted earliest arrival temporal path counting problem.
\section{Structural Results}\label{section:structuralresults}
We show that it is possible to find an efficient $(s,z)$-path for the Min-Cost Earliest Arrival Path Problem (\probas{}) in a graph $G$, if we are able to solve the Min-Cost Fastest Path Problem (\probbs).
We use a transformed graph $G'$, in which a new source vertex and a single edge is added. %
The reduction is from a search problem to another search problem.
We show that it preserves the existence of solutions and we also provide a mapping between the solutions.
This is also known as \emph{Levin} reduction. %
\begin{lemma}\label{lemma:levinred}
	There is a Levin reduction from \probas{} to \probbs. %
\end{lemma}
\begin{proof}
	Let $I=(G=(V,E),s,z)$ be an instance of \probas. We construct the \probbs{} instance $I'=(G'=(V\cup\{s'\}, E\cup \{e_0=(s',s,0,0,0)\}), s',z)$.
	Furthermore, let $\mathcal{X}_I$ be the sets of all $(s,z)$-paths for $I$, and $\mathcal{X}_{I'}$ be the sets of all $(s',z)$-paths for $I'$. 
	We define $g:\mathcal{X}_I\rightarrow \mathcal{X}_{I'}$ as bijection that prepends edge $e_0$ to the paths in $\mathcal{X}_I$, i.e., $g((e_1,\ldots,e_k))=(e_0,e_1,\ldots,e_k)$. 
	We show that $P\in \mathcal{X}_I$ is efficient (for \probas) iff $g(P)\in \mathcal{X}_{I'}$ is efficient (for \probbs{}). 
	
	Let $P=(e_1,\ldots,e_k)$ be an efficient $(s,z)$-path in $G$ w.r.t. costs and arrival time.
	Then $Q\coloneqq g(P)=(e_0,e_1,\ldots,e_k)$ is an $(s',z)$-path in $G'$ with $a(Q)=a(P)$ and $c(Q)=c(P)$.
	Now, assume $Q$ is not efficient w.r.t. costs and duration in $G'$. 
	Then there is a path $Q'$ with less costs and at most the duration of $Q$ or with shorter duration and at most the same costs of $Q$.
	Path $Q'$ also begins with edge $e_0$, and $G$ contains a path $P'$ that uses the same edges as $Q'$ with exception of edge $e_0$.
	Then, at least one of the following two cases holds.
	\begin{itemize}
		\item Case $c(Q')\leq c(Q)$ and $d(Q')<d(Q)$: Since the costs of $e_0$ are $0$, it follows that $c(P')=c(Q')\leq c(Q)=c(P)$.
		Because the paths start at time $0$ and for each path $d(P)=a(P)-s(P)$,
		it follows $d(Q')=a(Q')=a(P')<a(P)=a(Q)=d(Q)$. 
		\item Case $c(Q')<c(Q)$ and $d(Q')\leq d(Q)$: 
		Analogously, here we have $c(P')=c(Q')<c(Q)=c(P)$. And $a(P')=a(Q')\leq a(Q)=a(P)$.
	\end{itemize}
	Either of these two cases leads to a contradiction to the assumption that $P$ is efficient.

	Let $Q=(e_0,e_1,\ldots,e_k)$ and assume it is an efficient $(s',z)$-path in $G'$ w.r.t. to cost and duration. %
	Then there exists an $(s,z)$-path $P=(e_1,\ldots,e_k)$ in $G$ such that $g(P)=Q$, $a(Q)=a(P)$ and $c(Q)=c(P)$.
	Now, assume that $P$ is not efficient. %
	Then there is a path $P'$ with less costs and not later arrival time than $P$ or with earlier arrival time and at most the costs of $P$.
	In $G'$ exists the path $Q'=g(P')$ that uses the same edges as $P'$, and additionally the edge $e_0$ as \prefixp{} from $s'$ to $s$.
	We have the cases $c(P')<c(P)$ and $a(P')\leq a(P)$ or $c(P')\leq c(P)$ and $a(P')<a(P)$. Again, either of them leads to a contradiction to the assumption that $Q$ is efficient.
\end{proof}
Based on this result, we first present an algorithm for \probbenum{} in \Cref{section:algprobb} that we use in a modified version to solve \probaenum{} in \Cref{section:algproba}. 
In the rest of this section, we focus on graphs with strictly positive edge costs. 
\begin{observation}\label{observation:posweightsimple}
	Let $G=(V,E)$ be a weighted temporal graph.
	If for all edges $e=(u,v,t,\lambda, c)\in \tge$ it holds that $c>0$, then all efficient walks for \probaenum{} and \probbenum{} are simple, i.e., are paths.
\end{observation}
Similar to the non-temporal static case, it is possible to delete the edges of a cycle contained in the non-simple walk.
We denote the two special cases for graphs with strictly positive edge costs by \probbenumsp{} and \probaenumsp.
Our enumeration algorithms use a label setting technique. 
A label $l=(b,a,c,p,v,r,\Pi)$ at vertex $v\in V$ corresponds to an $(s,v)$-path and consists of the following entries:
\begin{center}
    \begin{tabular}{c@{\hspace{20mm}}l}
       $b=s(P_{s,v})$ & starting time, \\
       $a=a(P_{s,v})$ & arrival time, \\
       $c=c(P_{s,v})$ & cost, \\
       $p$ & predecessor label, \\
       $v$ & current vertex, \\
       $r$ & availability time of the previous edge and \\
       $\Pi$ & reference to a list of equivalent labels.
    \end{tabular}
\end{center}
Moreover, each label is uniquely identifiable by an additional identifier and has a reference to the edge that lead to its creation (denoted by $l.edge$). 
The proposed algorithms process the edges in order of their availability time. 
When processing an edge $e=(u,v,t,\lambda,c_e)$, all paths that end at vertex $u$ can be extended by pushing labels over edge~$e$ to vertex~$v$. Pushing a label $l=(b,a,c,p,u,r,\Pi)$ over $e$ means that we create a new label $l_{new}=(b,t+\lambda,c+c_e,l,v,t,\cdot)$ at vertex~$v$.

If we would create and store a label for each efficient path, we possibly would need exponential space in the size of the input. 
The reason is that the number of efficient paths can be exponential in the size of the input. 
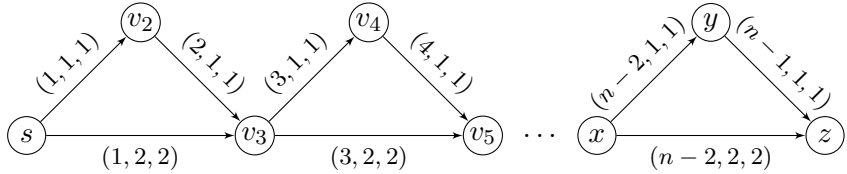
\begin{figure}[ht]
	\centering
	\begin{tikzpicture}
	\node[vertex] (1) at (-3,-1.5) {$s$};
	\node[vertex] (2) at (0,-1.5)  {$v_3$};
	\node[vertex] (3) at (-1.5,0)  {$v_2$};
	\node[vertex] (4) at (4.5,-1.5)  {$x$};
	\node[vertex] (5) at (1.5,0) {$v_4$};
	\node[vertex] (6) at (7.5,-1.5) {$z$};
	\node[vertex] (7) at (6,0) {$y$};
	\node[vertex] (8) at (3,-1.5)  {$v_5$};	
	
	\node at (3.75,-1.5)  {$\ldots$};	
	
	\path[->] 
	(1)  edge[dedge]  node[sloped,above,midway] {$(1,1,1)$}   (3)
	(1)  edge[dedge]  node[midway,below]  {$(1,2,2)$}  (2)
	(3)  edge[dedge]  node[sloped,above,midway] {$(2,1,1)$}   (2)		
	
	(2)  edge[dedge]  node[sloped,above,midway] {$(3,1,1)$}   (5)
	(2)  edge[dedge]  node[midway,below] {$(3,2,2)$}   (8)
	(5)  edge[dedge]  node[sloped,above,midway] {$(4,1,1)$}   (8)
	
	(4)  edge[dedge]  node[midway,below] {$(n-2,2,2)$}   (6)
	(4)  edge[dedge]  node[sloped,above,midway] {$(n-2,1,1)$}   (7)
	(7)  edge[dedge]  node[sloped,above,midway] {$(n-1,1,1)$}   (6);	
	\end{tikzpicture}
	\caption{Example for an exponential number of efficient paths for \probaenum{} and \probbenum. All $(s,v)$-paths for $v\in V$ are efficient.}
	\label{fig:exp_eff_path_01}
\end{figure}
\Cref{fig:exp_eff_path_01} shows an example for \probbenum{} and \probaenum, which
is similar to the one provided by Hansen~\cite{HansenBSP}, but adapted to the weighted temporal case.
$G$ has $n=\frac{2m}{3}+1$ vertices and $m>3$ edges. There are two paths from $s$ to $v_3$, four paths from $s$ to $v_5$, eight paths from $s$ to $v_7$ and so on. All $(s,v)$-paths for $v\in V$ are efficient. In total, there are $2^{\lfloor\frac{n}{2}\rfloor}$ efficient $(s,z)$-paths to be enumerated. 
However, the following lemma shows properties of the problems that help us to achieve polynomial time delay and a linear or polynomial space complexity.
Let $\mathcal{Y}_{A}$ ($\mathcal{Y}_{F}$) denote the objective space for \probaenum{} (\probbenum{}, respectively).
Moreover, we define $\mathcal{S}$ to be the number of distinct availability times of edges leaving the source vertex $s$.
\begin{lemma}\label{lemma:prop1}
	For \probaenum{}, the number of nondominated points in $\mathcal{Y}_A$ is in $\mathcal{O}(m)$. 
	For \probbenum{}, the number of nondominated points in $\mathcal{Y}_F$ is in $\mathcal{O}(\mathcal{S}\cdot m)=\mathcal{O}(m^2)$.
\end{lemma}
\begin{proof}
	Let $G=(V,E)$ be a weighted temporal graph and $s,z\in V$.
	First we show that the statement holds for \probaenum.
	The possibilities for different arrival times at vertex $z$ is limited by the number of incoming edges at $z$. For each path $P_{s,z}$ we have 
	\[a(P_{s,z})\in\{\alpha\mid \alpha\coloneqq t_e+\lambda_e \text{ with } e=(u,z,t_e,\lambda_e,c) \in \delta^-(z)\}\text{.}\]
	Consequently, there are at most $|\delta^-(z)|\in\mathcal{O}(m)$ different arrival times. 
	For each arrival time $a$, there can only be one nondominated point $(a,c)\in \mathcal{Y}_A$ that has the minimum costs of $c$, and which represents exactly all efficient paths with arrival time $a$ and costs $c$.
	
	Now consider the case for \probbenum. The number of distinct availability times of edges leaving the source vertex $\mathcal{S}$ is bounded by $|\delta^+(s)|\in \mathcal{O}(m)$.
	Because the duration of any $(s,z)$-path $P$ equals $a(P)-s(P)$, there are at most $\mathcal{S}\cdot |\delta^-(z)|\in\mathcal{O}(m^2)$ different durations possible at vertex $z$.
	For each duration, there can only be one nondominated point $(d,c)\in \mathcal{Y}_F$ having minimum costs $c$. %
\end{proof}
Note that for general bicriteria optimization (path) problems there can be an exponential number of nondominated points in the objective space. Skriver and Andersen~\cite{SKRIVER2000507} give an example for a family of graphs with an exponential number of nondominated points for a bicriteria path problem.
The fact that in our case the number of nondominated points in the objective space is polynomially bounded, allows us to achieve polynomial time delay and space complexity for our algorithms. 
The idea is to consider equivalence classes of labels at each vertex, such that we only have to proceed with a single representative for each class.
First, we define the following relations between labels.
\begin{definition}\label{definition:labeldom2}
	Let $l_1=(b_1,a_1,c_1,p_1,v,r_1,\Pi_1)$ and $l_2=(b_2,a_2,c_2,p_2,v,r_2,\Pi_2)$ be two labels at vertex $v$.
	\begin{enumerate}
		\item Label $l_1$ is \emph{equivalent} to $l_2$ iff $c_1=c_2$ and $b_1=b_2$. 
		\item Label $l_1$ \emph{predominates} $l_2$ if $l_1$ and $l_2$ are not equivalent, $b_1 \geq b_2$, $a_1 \leq a_2$ and $c_1 \leq c_2$ with at least one of the inequalities being strict.  
		\item Finally, label $l_1$ \emph{dominates} $l_2$ if $a_1-b_1 \leq a_2-b_2$ and $c_1 \leq c_2$ with at least one of the inequalities being strict.  
	\end{enumerate}
\end{definition}
For each class of equivalent labels, we have a representative $l$ and a list $\Pi_l$ that contains all equivalent labels to $l$. 
For each vertex $v\in V$, we have a set $R_v$ that contains all representatives. %
The algorithms consist of two consecutive phases:
\begin{itemize}
	\item \emph{Phase 1} calculates the set of non-equivalent representatives $R_v$ for every vertex $v\in V$ such that every label in $R_v$ represents a set of equivalent paths from $s$ to $v$. 
	For each of the nonequivalent labels $l\in R_v$ we store the list $\Pi_l$ that contains all labels equivalent to $l$. 
	\item \emph{Phase 2} recombines the sets of equivalent labels in a backtracking fashion, such that we are able to enumerate exactly all efficient $(s,z)$-paths without holding the paths in memory.
\end{itemize}
A label in $\Pi_l$ at vertex $v$ represents all $(s,v)$-paths that are extension of all paths represented by its predecessor, and
$l\in R_v$ is a representative for all labels in $\Pi_l$. The representative itself is in $\Pi_l$ and has minimum arrival time among all labels in $\Pi_l$.

We have to take into account that a \prefixp{} $P_{s,w}$ of an efficient $(s,z)$-path may not be an efficient $(s,w)$-path. 
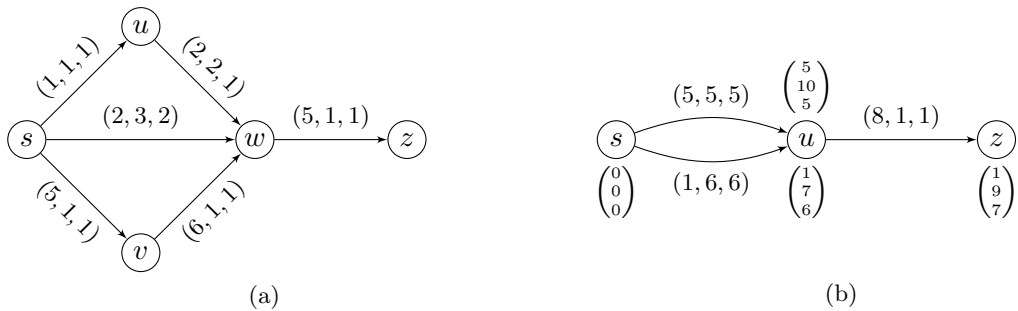
\begin{figure}[ht]
	\subfloat[]{ %
		\begin{minipage}{0.46\textwidth}
		\begin{tikzpicture}
		\node[vertex] (1) at (-2.5,0.5) {$u$};
		\node[vertex] (2) at (-1,-1) {$w$};
		\node[vertex] (3) at (1.0,-1) {$z$};
		\node[vertex] (4) at (-2.5,-2.5) {$v$};
		\node[vertex] (5) at (-4,-1) {$s$};			
		
		\path[->] 
		(1)  edge[dedge]  node[midway,above,sloped] {$(2,2,1)$}   (2)
		(5)  edge[dedge]  node[midway,above,sloped] {$(2,3,2)$}   (2)	
		(4)  edge[color=black,dedge]  node[midway,below,sloped] {$(6,1,1)$}   (2)
		(2)  edge[color=black,dedge]  node[midway,above,sloped] {$(5,1,1)$}   (3)
		(5)  edge[dedge]  node[midway,above,sloped] {$(1,1,1)$}   (1)
		(5)  edge[color=black,dedge]  node[midway,below,sloped] {$(5,1,1)$}   (4);	
		\end{tikzpicture}
	\end{minipage}
	}\subfloat{\begin{minipage}{.03\textwidth}\hfill\end{minipage}}
    \setcounter{subfigure}{1}
	\subfloat[]{ %
		\begin{minipage}{0.46\textwidth}
	\begin{tikzpicture}
	\node at (2.5,-3.1) {};
	\node at (2.5,.1) {};
	\node at (2.5,-2.2) {$\tiny\begin{pmatrix}0\\0\\0\end{pmatrix}$};	
	\node at (5,-0.8) {$\tiny\begin{pmatrix}5\\10\\5\end{pmatrix}$};	
	\node at (5,-2.2) {$\tiny\begin{pmatrix}1\\7\\6\end{pmatrix}$};			
	\node at (7.5,-2.2) {$\tiny\begin{pmatrix}1\\9\\7\end{pmatrix}$};				
	\node[vertex] (s) at (2.5,-1.5) {$s$};			
	\node[vertex] (u) at (5,-1.5) {$u$};
	\node[vertex] (z) at (7.5,-1.5) {$z$};
	
	\path[->] 
	(s)  edge[dedge, bend left=20]  node[midway,above,sloped] {$(5,5,5)$}   (u)
	(s)  edge[dedge, bend right=20]  node[midway,below,sloped] {$(1,6,6)$}   (u)
	(u)  edge[dedge]  node[midway,above,sloped] {$(8,1,1)$}   (z);
	
	\end{tikzpicture}
		\end{minipage}
		}
	\caption{(a) An example for non-efficient \prefixp{}s. (b) The vertices are annotated with labels that describe the starting time, arrival time and costs of the paths starting at $s$.}
	\label{fig:counter_example_optimality_principle_01}
\end{figure}
\Cref{fig:counter_example_optimality_principle_01} (a) shows an example for a weighted temporal graph with a non-optimal \prefixp{}. Consider the following paths: 
\begin{itemize}
	\item $P^*_{s,z}=((s,w,2,3,2),(w,z,5,1,1))$ with arrival time $6$ and duration $4$
	\item $P^1_{s,w}=((s,w,2,3,2))$ with arrival time $5$ and duration $3$
	\item $P^2_{s,w}=((s,u,1,2,1),(u,w,2,1,1))$ with arrival time $3$ and duration $3$
	\item $P^3_{s,w}=((s,v,5,1,1),(v,w,6,1,1))$ with arrival time $7$ and duration $2$
\end{itemize}
All $(s,z)$-paths have cost $3$ and all $(s,w)$-paths have cost $2$. Path $P^*_{s,z}$ is efficient for \probbenum{} and \probaenum{}. For \probaenum{} the \prefixp{} $P^1_{s,w}$ is not efficient, because $P^2_{s,w}$ arrives earlier. However, for \probbenum{} the only efficient $(s,w)$-path is $P^3_{s,w}$. Consequently, we cannot discard a non-efficient path that possibly is a \prefixp{} of an efficient path. 
We use the predomination relation to remove all labels that do not represent \prefixp{s} of efficient paths.
\begin{lemma}\label{lemma:domination1}
	Let $l_1=(b_1,a_1,c_1,p_1,v,r_1,\Pi_1)$ and $l_2=(b_2,a_2,c_2,p_2,v,r_2,\Pi_2)$ be two distinct labels at vertex $v\in V$.
	If $l_1$ predominates $l_2$, then $l_2$ cannot be a label representing a \prefixp{} of any efficient path. 
\end{lemma}
\begin{proof}
	There are two distinct paths $P^1_{s,v}$ and $P^2_{s,v}$ from $s$ to $v$ corresponding to $l_1$ and $l_2$. 
	Due to the predomination of $l_1$ over $l_2$ it follows that $a_1=a(P^1_{s,v})\leq a(P^2_{s,v})=a_2$, $b_1=s(P^1_{s,v})\geq s(P^2_{s,v})=b_2$ and $c_1=c(P^1_{s,v})\leq c(P^2_{s,v})=c_2$ with at least one of the later two relations being strict, due to the fact that the labels are not equivalent.
	Let $P_{s,w}$ be a path from $s$ to some $w\in V$ such that $P^2_{s,v}$ is \prefixp{} of $P_{s,w}$, and assume that $P_{s,w}$ is efficient.
	Let $P'_{s,w}$ be the path where the \prefixp{} $P^2_{s,v}$ is replaced by $P^1_{s,v}$. This is possible because $a(P^1_{s,v})\leq a(P^2_{s,v})$. %
	Now, since $s(P^1_{s,v})\geq s(P^2_{s,v})$ and $c(P^1_{s,v})\leq c(P^2_{s,v})$ with at least one of the inequalities being strict, it follows that $a(P'_{s,w})-s(P'_{s,w})\leq a(P_{s,w})-s(P_{s,w})$ and $c(P'_{s,w})\leq c(P_{s,w})$ also with one of the inequalities being strict. Therefore, $P'_{s,w}$ dominates $P_{s,w}$, a contradiction to the assumption that $P_{s,w}$ is efficient.	
\end{proof}
\Cref{fig:counter_example_optimality_principle_01} (b) shows an example for a non-predominated label of a \prefixp{} that we cannot discard. 
Although path $P_1=((s,u,5,5,5))$ dominates path $P_2=((s,u,1,6,6))$, we cannot discard $P_2$. 
The reason is that the arrival time of $P_1$ is later than the availability time of the only edge from $u$ to $z$. Therefore, $P_2$ is the \prefixp{} of the only efficient path $((s,u,1,6,6),(u,z,8,1,1))$.
\section{Min-Cost Fastest Path Enumeration Problem}\label{section:algprobb}
In this section, we present the algorithm for \probbenum.
\Cref{alg:alg2} expects as input a weighted temporal graph with strictly positive edge costs in the edge stream representation, the source vertex $s\in V$ and the target vertex $z\in V$.
\begin{algorithm}\mbox{\hfill}
	\\\textbf{Input:} Graph $\tg$ in edge stream representation, source $s\in V$ and target $z\in V$
	\\\textbf{Output:} All efficient $(s,z)$-paths
	\begin{algorithmic}[1]
		\Statex \textit{Phase 1}
		\State{initialize $R_v$ for each $v\in V$}
		\State{insert label $l_{init}=(0,0,0,-,s,-,\Pi_{l_{init}})$ into $R_s$ and $\Pi_{l_{init}}$}
		\For{each edge $e=(u,v,t_e,\lambda_e,c_e)$}\label{alg:alg_2:mainloop}
		\State $S\leftarrow \{(b,a,c,p,v,r,\cdot) \in R_u \mid a\leq t_e$, $c$ minimal and distinct starting times $b\}$\label{alg:alg_2:step:minset}
		\For{each $l=(b,a,c,p,v,r,\cdot) \in S$ with $a\leq t_e$}\label{alg:step:break1}
		\If{$u=s$}
		\State $l_{new}\leftarrow(t_e,t_e+\lambda_e,c_e,l,s,t_e,\cdot)$
		\Else
		\State $l_{new}\leftarrow(b,t_e+\lambda_e,c+c_e,l,u,t_e,\cdot)$
		\EndIf
		\For{each $l'=(b',a',c',p',v,t',\Pi')\in R_v$}\label{alg:step:dom}
		\If{$l_{new}$ predominates $l'$} 
		\State remove $l'$ from $R_v$ and delete $\Pi'$
		\ElsIf {$l'$ is equivalent to $l_{new}$} 
		\State insert $l_{new}$ into $\Pi'$
		\State set reference $\Pi\leftarrow \Pi'$		
		\If{$t_e+\lambda_e<a'$}
			\State replace $l'$ in $R_v$ by $l_{new}$
		\EndIf
		\State goto \ref{alg:step:break1}
		\ElsIf {$l'$ predominates $l_{new}$} 
		\State delete $l_{new}$ 
		\State goto \ref{alg:step:break1}
		\EndIf
		\EndFor
		\State insert $l_{new}$ into $R_v$ and initialize $\Pi_{l_{new}}$ with $l_{new}$ 
		\EndFor
		\EndFor
		\Statex \textit{Phase 2}
		\State mark nondominated labels in $R_z$\label{alg:step:dom_2} 
		\For{each  marked label $l'=(b,a,c,p,z,r,\Pi) \in R_z$}
		\For {each label $l\in \Pi$ with minimal arrival time} 
		\State initialize empty path $P$
		\State call \textsc{OutputPaths}($l$, $P$); %
		\EndFor
		\EndFor
		\Statex
		\Statex \textit{Procedure for outputting paths}		
		\Procedure{OutputPaths}{label $l=(b,a,c,p,cur,r,\Pi)$, path $P$} %
			\State prepend edge $l.edge$ to $P$ \label{alg:alg_earliest:step:opprepend}
			\If{$l$ has predecessor $p=(b_p,a_p,c_p,p_p,v_p,r_p,\Pi_p)$}
				\For {each label $l'=(b_{l'},a_{l'},c_{l'},p_{l'},v_{l'},r_{l'},\Pi_{p})$ in $\Pi_p$}
					\If{$a_{l'}\leq r$}\label{alg:alg_earliest:step:feasiblepath}
						\State call \textsc{OutputPaths}($l'$, $P$, $visited$)
					\EndIf
				\EndFor
				\State \Return
			\EndIf
			\State output path $P$\label{alg:alg_earliest:step:opout}
		\EndProcedure
	\end{algorithmic}		
	\caption{for \probbenum}
	\label{alg:alg2}
\end{algorithm}
First, we insert an initial label $l_{init}$ into $R_s$ and $\Pi_{l_{init}}$. The algorithm then successively processes the $m$ edges in order of their availability time.
For each edge $e=(u,v,t,\lambda,c)$, we first determine the set $S\subseteq R_u$ of labels with distinct starting times, minimal costs and an arrival time less or equal to $t$ at vertex $u$ (line \ref{alg:alg_2:step:minset}). Next, we push each label in $S$ over $e$.
We check for predomination and equivalence with the other labels in $R_v$ and discard all predominated labels. 
In case the new label is predominated, we discard it and continue with the next label in $S$. In case that the new label $l_{new}$ is equivalent to a label $l=(a,c,p,v,t_e,\Pi)\in R_v$, we add $l_{new}$ to $\Pi$.  If $l_{new}$ arrives earlier at $v$ than the arrival time of $l$, we replace the representative $l$ with $l_{new}$ in $R_v$. If the new label is not predominated and not equivalent to any label in $R_v$, we insert $l_{new}$ into $R_v$ and $\Pi_{l_{new}}$. In this case, $l_{new}$ is a new representative and we initialize $\Pi_{l_{new}}$ (which contains only $l_{new}$ at this point). %
For the following discussion, we define the set of all labels at vertex $v\in V$ as $L_v\coloneqq \bigcup\limits_{l\in R_v} \Pi_l$.
\begin{lemma}\label{lemma:effpathlables_alg2}
	Let $P_{s,v}$ be an efficient path and $P_{s,w}$ a \prefixp{} of $P_{s,v}$. %
	At the end of Phase 1 of \Cref{alg:alg2}, $R_w$ contains a label representing $P_{s,w}$.
\end{lemma}
\begin{proof}
	We show that each \prefixp{} $P_0,P_1,\ldots,P_{k}$, with $P_0=P_{s,s}$ and $P_k=P_{s,v}$ 
	is represented by a label at the last vertex of each \prefixp{} by induction over the length $h$. 
	Note that all \prefixp{}s have the same starting time $b=s(P_{s,v})$.
	For $h=0$ we have $P_0=P_{s,s}$ and since $s$ does not have any incoming edges, the initial label $l_{init}=l_0$ representing $P_0$ is in $L_s$ after Phase 1 finishes.
	Assume the hypothesis is true for $h= i-1$ and consider the case for $h=i$ and the \prefixp{} $P_{i}=P_{s,v_{i+1}}=(e_1,\ldots,e_i=(v_i,v_{i+1},t_i,\lambda_i,c_i))$,
	which consists of the \prefixp{} $P_{i-1}=P_{s,v_{i}}=(e_1,\ldots,e_{i-1}=(v_{i-1},v_{i},t_{i-1},\lambda_{i-1},c_{i-1}))$ and edge $e_i=(v_i,v_{i+1},t_i,\lambda_i,c_i)$.
	Due to the induction hypothesis, we conclude that $L_{v_i}$ contains a label $l_{i-1}=(b,a_{i-1},c_{i-1},p_{i-1},v_i,r_{i-1},\Pi_{i-1})$ that represents $P_{i-1}$.
	Because $P_{i-1}$ is a \prefixp{} of $P_{s,v}$ the representing label $l_{i-1}$ must have the minimum cost in $L_{v_i}$ under all labels with starting time $b$ before edge $e_i$ arrives.
	Else, it would have been predominated and replaced by a cheaper one (\Cref{lemma:domination1}). %
	The set $S$ contains a label that represents $l_{i-1}$, because the representative of $\Pi_{i-1}$ has an arrival time less or equal to $a_{i-1}$.
	Therefore, the algorithm pushes $l_{new}=(b,t_i+\lambda_i,c_{i-1}+c_i,l_{i-1},v_{i+1},t_i,\cdot)$ over edge $e_i$. 
	If $R_{v_{i+1}}$ is empty the label $l_{new}$ gets inserted into $R_{v_{i+1}}$ and $\Pi_{l_{new}}$.
	Otherwise we have to check for predomination and equivalence with every label $l'=(b',a',c',p',v_{i+1},r',\Pi')\in R_{v_{i+1}}$. There are the following cases:
	\begin{enumerate}
		\item $l_{new}$ predominates $l'$: We can remove $l'$ from $R_{v_{i+1}}$ because it will never be part of an efficient path (\Cref{lemma:domination1}). The same is true for each label in $\Pi'$ and therefore we delete $\Pi'$. 
		However, we keep $l_{new}$ and continue with the next label.
		\item $l_{new}$ and $l'$ are equivalent: We add $l_{new}$ to $\Pi'$. In this case we represent the path $P_i$ by the representative of $\Pi'$.
		Consequently, the path is represented by a label in $L_{v_{i+1}}$.
	\end{enumerate}	
	If neither of these two cases apply for any label in $R_{v_{i+1}}$, we add $l_{new}$ to $R_{v_{i+1}}$ and to $\Pi_{l_{new}}$. 
	The case that a label $l$ is not equivalent to $l_{new}$ and predominates $l_{new}$ cannot be for the following reason. 
	If $l$ predominates $l_{new}$, there is a path $P'$ from $s$ to $v_{i+1}$ with less costs or later starting time (because $l$ and $l_{new}$ are not equivalent) and a not later arrival time. 
	Replacing the \prefixp{} $P_{i}$ with $P'$ in the path $P_{s,v}$ would lead to an $(s,v)$-path with less costs and/or shorter duration. 
	This contradicts our assumption that $P_{s,v}$ is efficient.
	Therefore, after Phase 1 finished, the label $l_{new}$ representing the \prefixp{} $P_{i}$ is in $L_{v_i}$.
	It follows that if $P_{s,v}=(e_1,\ldots,e_k)$ is an efficient path, then after Phase 1 the set $R_v$ contains a label representing $P_{s,v}$ 
	(possibly, such that a label in $R_v$ represents a list of equivalent labels, that contains the label representing $P_{s,v}$).
\end{proof}
After all edges have been processed, the algorithm continues with Phase 2. 
First, the algorithm marks all nondominated labels in $R_z$. For each marked label $l$ the algorithm iterates over the list of equivalent labels $\Pi_l$ and calls the output procedure for each label in $\Pi_l$. 
We show that all and only efficient paths are enumerated.
\begin{theorem}
	Let $G=(V,E)$ be a weighted temporal graph with strictly positive edge costs and $s,z\in V$ an instance of \probbenum{}. Algorithm \ref{alg:alg2} outputs exactly all efficient $(s,z)$-paths. %
\end{theorem}
\begin{proof}
	\Cref{lemma:effpathlables_alg2} implies that for each efficient path $P_{s,z}$ there is a corresponding representative label in $R_z$ after Phase 1 is finished. 
	Note that there might also be labels in $L_z$ that do not represent efficient paths. 
	First, we mark all nondominated labels in $R_z$.
	For every marked representative $l'=(b,a,c,p,z,r,\Pi_{l'})$ in $R_z$ we proceed by calling the output procedure for all labels $l\in \Pi_{l'}$ with minimal arrival time.
	Each such label $l$ represents at least one efficient $(s,z)$-path and we call the output procedure with $l$ and the empty path $P$. 
	Let path $Q=(e_1\ldots,e_k)$ be an efficient $(s,z)$-paths represented by $l$. 
	We show that the output procedure successively constructs the \suffixp{}s $P_{k-i+1}=(e_{k-i+1},\ldots,e_{k})$ of $Q$ for $i\in\{1,\ldots,k\}$ and finally outputs $Q=P_1=(e_1,\ldots,e_{k})$. 

	We use induction over the length $i\geq 1$ of the \suffixp{}. For $i=1$ the statement is true. 
	$P_k=(e_k)$ is constructed by the first instruction which prepends the last edge of $Q$ to the initially empty path. %
	Assume the statement holds for a fixed $i<k$, i.e., the \suffixp{} $P_{k-i+1}$ of $Q$ with length $i$ has been constructed, by calling the output procedure with $P_{k-i+2}$ and label $\tilde{l}=(b,a,c,p,v_{k-i+1},r,\Pi)$. 
	Now, for $i+1$, the \suffixp{} $P_{k-i}=(e_{k-i},\ldots,e_{k})$ equals suffix-path $P_{k-i+1}$ with the additional edge $e_{k-i}=(v_{k-i},v_{k-i+1},t,\lambda,c)$ prepended and where $v_{k-i+1}$ is the first vertex of $P_{k-i+1}$.
	If label $\tilde{l}$ has a predecessor $p=(b,a_p,c_p,p_p,v_{k-i},r_p,\Pi')$,
	the algorithm recursively calls the output procedure for each label in the list of equivalent labels $\Pi'$ %
	that has an arrival time less or equal to the time $r$ of the edge that lead to the creation of $\tilde{l}$.
	Due to \Cref{lemma:effpathlables_alg2} and the induction hypothesis, the algorithm particularly calls the output procedure for the label $l'$ that represents the beginning of the \suffixp{} $P_{k-i}$.
	Because $Q$ is a temporal path the arrival time at $v_{k-i}$ is less or equal to the time of the edge that lead to the label $\tilde{l}$. 
	Consequently, there is a call of the output procedure that constructs $P_{k-i}$.
	If label $\tilde{l}$ does not have a predecessor, the algorithm arrived at vertex $s$ and the algorithm outputs the complete path $Q=P_1$.
	Therefore, all efficient paths are enumerated.
		
	We still have to show that only efficient paths are enumerated.	In order to enumerate a non-efficient $(s,z)$-path $Q'$, there has to be a label $l_q$ in $L_z$ for which the output procedure is called and which represents $Q'$.
	For $Q'$ to be non-efficient there has to be at least one label $l_d$ in $L_z$ that dominates $l_q$. 
	In line~\ref{alg:step:dom_2} the algorithm marks all nondominated labels in $R_z$. This implies that $l_d$ and $l_q$ have the same cost and starting times and that they are in the same list, let this list be $\Pi_x$ for some label $x\in R_z$. 
	Because $l_q$ is dominated by $l_d$ the arrival time of $l_d$ is strictly earlier than the the arrival time of $l_q$.
	However, we call the output procedure only for the labels in $\Pi_x$ with the minimal arrival time.
	Consequently, it is impossible that the non-efficient path $Q'$ is enumerated.
	Finally, because all edge costs are strictly positive and due to \Cref{observation:posweightsimple} only paths are enumerated. 
\end{proof}
\paragraph{Example:} \Cref{fig:example_alg2} shows an example for \Cref{alg:alg2} at the end of Phase~1. The indices of the edges are according to the position in the sequence of the edge stream.
The representative labels at the vertices only show the starting time, arrival time and cost.
The lists $\Pi$ of equivalent labels are not shown. All of them contain only the representative, with exception of $\Pi_l$ represented by label $l$ in $R_w$.
The list $\Pi_l$ contains label $l=(3,7,4)^T$ representing path $((s,w,3,4,4))$ and the equivalent label $(3,8,4)^T$ representing path $((s,u,3,3,3),(u,w,6,1,1))$. 
There are three efficient paths. Starting the output procedure from vertex $z$ with the label $(7,10,6)^T$ yields path $(e_5,e_8)$,
and starting with label $(3,9,5)^T$ yields the two paths $(e_1,e_4,e_7)$ and $(e_2,e_7)$.
Notice that label $(7,9,2)^T$ in $R_w$ which dominates label $(3,7,4)^T$, is not part of an efficient $(s,z)$-path, due to its late arrival time.
\begin{figure}
	\centering
	\scalebox{0.94}{
	\begin{tikzpicture}
	\node[vertex] (s) at (1,0) {s};
	\node[vertex] (u) at (4.5,0) {u};
	\node[vertex] (v) at (6,4) {v};
	\node[vertex] (w) at (7.5,0) {w};
	\node[vertex] (z) at (11,0) {z};			
	
	\node at (0.8,0.9) {$R_s=\left[\tiny\begin{pmatrix}0\\0\\0\end{pmatrix}\right]$};
	\node at (4,-0.7) {$\tiny R_u=\left[\begin{pmatrix}3\\6\\3\end{pmatrix}\right]$};	
	\node at (4.5,4) {$\tiny R_v=\left[\begin{pmatrix}7\\8\\1\end{pmatrix}\right]$};
	\node at (8.8,-0.7) {$\tiny R_w=\left[l=\begin{pmatrix}3\\7\\4\end{pmatrix},\begin{pmatrix}7\\9\\2\end{pmatrix}\right]$};
	\node at (11.5,0.8) {$\tiny R_z=\left[\begin{pmatrix}7\\10\\6\end{pmatrix},\tiny\begin{pmatrix}3\\9\\5\end{pmatrix}\right]$};	
	
	\path[->] 
	(s)  edge[dedge]  node[midway,above,sloped] {$e_1=(3,3,3)$}   (u)
	(s)  edge[dedge]  node[midway,above,sloped] {$e_5=(7,1,1)$}   (v)
	(s)  edge[dedge,bend right=50]  node[midway,below,sloped] {$e_2=(3,4,4)$}   (w)	
	(u)  edge[dedge]  node[midway,above,sloped] {$e_3=(6,3,1)$}   (v)
	(u)  edge[dedge]  node[midway,above,sloped] {$e_4=(6,2,1)$}   (w)
	(v)  edge[dedge]  node[midway,above,sloped] {$e_6=(8,1,1)$}   (w)
	(v)  edge[dedge]  node[midway,above,sloped] {$e_8=(9,1,5)$}   (z)
	(w)  edge[dedge]  node[midway,above,sloped] {$e_7=(8,1,1)$}   (z);	
	\end{tikzpicture}}
	\caption{Example for \Cref{alg:alg2}. Each vertex is annotated with the representatives after Phase 1 finished.} %
	\label{fig:example_alg2}
\end{figure}
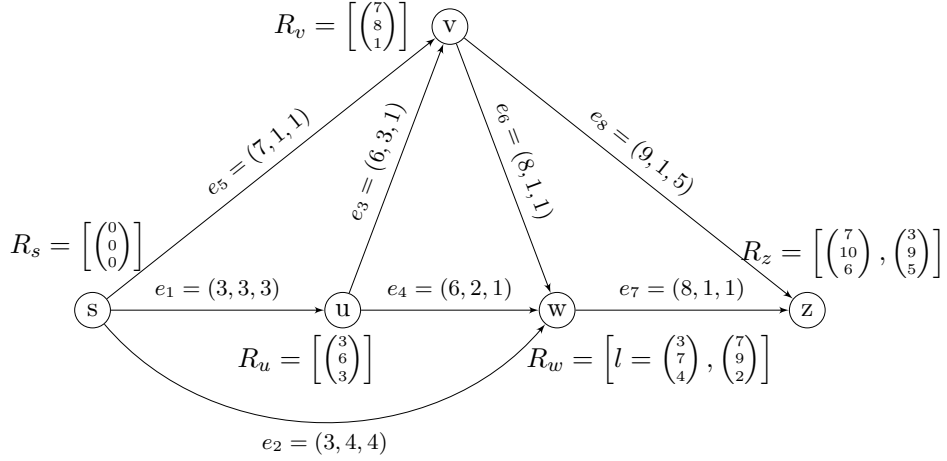

\begin{lemma}\label{lemma:alg2_timephase1}
	Phase 1 of \Cref{alg:alg2} has a time complexity of $\mathcal{O}(\mathcal{S}\cdot m^2)$.
\end{lemma}
\begin{proof}
	The outer loop iterates over $m$ edges.
	For each edge $e=(u,v,t_e,\lambda_e,c_e)$ we have to find the set $S\subseteq R_u$ consisting of all labels with minimal cost, distinct starting times and arrival time less or equal to $t_e$ (see line \ref{alg:alg_2:step:minset}). This can be done in $\mathcal{O}(m)$ time.
	For each label in $S$ we have to check for predominance or equivalence with each label in $R_v$ in $\mathcal{O}(S\cdot m)$ total time.
	Since we have $|S|\leq \mathcal{S}$, we get a total time of $\mathcal{O}(\mathcal{S}\cdot m^2)$.
\end{proof}
The following lemma shows that the number of labels is polynomially bounded in the size of the input.
\begin{lemma}\label{lemma:space_labels_alg2}
	The total number of labels generated and hold at the vertices in \Cref{alg:alg2} is less than or equal than $\mathcal{S}\cdot m+1$.
\end{lemma}
\begin{proof}
	We need one initial label $l_{init}$.
	For each incoming edge $e=(u,v,t,\lambda,c)$ in the edge stream we generate at most $|S|\leq \mathcal{S}$ new labels which we push over $e$ to vertex $v$.
	Therefore, we generate at most $\mathcal{S}\cdot m+1$ labels in total.
\end{proof}
\begin{theorem}
	\probbenumsp{}  $\in$ \psdelayp.
\end{theorem}
\begin{proof}
	Phase 1 takes only polynomial time in size of the input, i.e., number of edges (\Cref{lemma:alg2_timephase1}). 
	In Phase 2 of \Cref{alg:alg2}, we first %
	find and mark all nondominated labels in $R_z$ in $\mathcal{O}(m^2)$ time. %
	For each nondominated label, we call the output procedure which visits at most $\mathcal{O}(m^2)$ labels and outputs at least one path.
	It follows that the time between outputting two consecutively processed paths is also bounded by $\mathcal{O}(m^2)$.
	Therefore, \probbenumsp{} is in \delayp{}.
	The space complexity is dominated by the number of labels we have to manage throughout the algorithm. Due to \Cref{lemma:space_labels_alg2}, the number of labels is in $\mathcal{O}(m^2)$.
	Consequently, \probbenumsp{} is in \psdelayp{}.
\end{proof}
Notice that if we allow zero-weighted edges the algorithm enumerates walks. 
However, removing zero-weighted cycles to obtain only paths could repeatedly lead to same path, such that we would not be able to guarantee that the time between outputting two successive efficient paths is polynomially bounded. However, the following problems are easy to decide, even if we allow zero weighted edges.
\begin{theorem}\label{theorem:algphase2poly}
	Given a weighted temporal graph $G=(V,E)$, $s,z\in V$ and
	\begin{enumerate}
		\item an $(s,z)$-path $P$, deciding if $P$ is efficient for \probbenum{}, or
		\item $c\in \mathbb{R}_{\geq 0}$ and $d\in \mathbb{N}$, deciding if there exists an $(s,z)$-path $P$ with $d(P)\leq d$ and $c(P)\leq c$
		is possible in polynomial time. 
	\end{enumerate}		
\end{theorem}
\begin{proof}
	We use Phase 1 of \Cref{alg:alg2} and calculate the set $\mathcal{N}\subseteq \mathcal{Y}$ of nondominated points.
	Due to the possibility of edges with cost $0$, there may be non-simple paths, i.e., walks, that have zero-weighted cycles. 
	Nonetheless, Phase 1 terminates after processing the $m$ edges.
	If there exists an efficient $(s,z)$-walk $W$ with $(c(W),a(W))$, then there also exists a simple and efficient $(s,z)$-path $Q$ with the same cost vector. $Q$ is the same path as $W$ but without the zero-weighted cycles.
	In order to decide if the given path $P$ is efficient, we first calculate the cost vector $(c(P),a(P))$, and then validate if $(c(P),a(P))\in \mathcal{N}$. 
	For \emph{2.}, we only need to compare $(c,d)$ to the points in $\mathcal{N}$.
	The size of $\mathcal{N}$ is polynomially bounded (\Cref{lemma:prop1}). Phase 1 and calculating the cost of $P$ takes polynomial time.
\end{proof}
We can find a maximal set of efficient paths with pairwise different cost vectors in polynomial time.
\begin{corollary}\label{corollary:noneffpoly}
	Given a temporal graph $G=(V,E)$ and $s,z\in V$, a 
	maximal set of efficient $(s,z)$-paths with pairwise different cost vectors for \probbenum{}
	can be found in $\mathcal{O}(\mathcal{S}\cdot m^2)$. 
\end{corollary}
\begin{proof}
	We use Phase 1 of \Cref{alg:alg2} and calculate the set $\mathcal{N}\subseteq \mathcal{Y}$ of nondominated points in $\mathcal{O}(\mathcal{S}\cdot m^2)$ time.
	Furthermore, we use a modified output procedure, that stops after outputting the first path. We call the procedure for each nondominated label in $R_z$, and if a walk is found we additionally
	remove all zero-weighted cycles. 
	Finding the walk and removing the cycles is possible in linear time, since the length of a walk is bounded by $m$. %
\end{proof}
\section{Min-Cost Earliest Arrival Path Enumeration Problem}\label{section:algproba}
Based on the reduction of \Cref{lemma:levinred} presented in \Cref{section:structuralresults}, we modify \Cref{alg:alg2} to solve \probaenumsp. 
Let $(G=(V,E),s,z)$ be the instance for \probaenumsp{} and $(G'=(V',E'),s',z)$ be the transformed instance in which all paths start at time $0$ at the new source $s'$. %
Although edge $(s,s',0,0,0)$ has costs $0$, because $s'$ has no incoming edges any efficient walk in the transformed instance is simple, i.e., a path.
With all paths starting at $0$, there are the following consequences for the relations between labels defined in \Cref{definition:labeldom2}. First, consider the \emph{equivalence}   
and let $l_1=(0,a_1,c_1,p_1,v,r_1,\Pi_1)$ and $l_2=(0,a_2,c_2,p_2,v,r_2,\Pi_2)$ be two labels at vertex $v$.
Because the starting time of both labels is $0$, the labels are equivalent if $c_1=c_2$.
It follows, that label $l_1$ \emph{predominates} $l_2$ if $a_1 \leq a_2$ and $c_1 < c_2$, hence there is no distinction between domination and predomination.

\Cref{alg:alg_earliest} shows a modified version of \Cref{alg:alg2}, that sets all starting times to $0$. 
The modified algorithm only needs a linear amount of space and less running time for Phase 1.
The reasons for this are, that in line \ref{alg:alg_earliest:step:select} it only needs to find a single label with the minimum costs, instead of the set ${S}$. %
And at each vertex $v$ it only has one representative $l$ in $R_v$ with minimal costs (w.r.t. the other labels in $R_v$), due to the equivalence of labels that have the same costs.
Furthermore, in Phase 2 we do not need to explicitly find the nondominated labels in $R_z$.
Because each label $l$ in $R_z$ has a unique cost value, we consider each represented class $\Pi_l$ and call the output procedure with the labels that have the minimum arrival time in $\Pi_l$.
\begin{algorithm}[thb]\mbox{\hfill}
	\\\textbf{Input:} Graph $\tg$ in edge stream representation, source $s\in V$ and target $z\in V$
	\\\textbf{Output:} All efficient $(s,z)$-paths
	\begin{algorithmic}[1]
		\Statex \textit{Phase 1}
		\State{initialize $R_v$ for each $v\in V$}
		\State{insert label $l_{init}=(0,0,0,0,-,s,\Pi_{l_{init}})$ into $R_s$ and $\Pi_{l_{init}}$}
		\For{each edge $e=(u,v,t_e,\lambda_e,c_e)$} \label{alg:alg_earliest:step:break1}
			\State $l\leftarrow(0,a,c,p,u,\cdot,\cdot) \in R_u$ with $a\leq t_e$ and $c$ minimal\label{alg:alg_earliest:step:select}
			\State $l_{new}\leftarrow(0,t_e+\lambda_e,\,c+c_e,\,l,\,v,\,t_e,\,\Pi)$
			\For{each $l'=(0,a',c',p',v,r',\Pi')\in R_v$}\label{alg:alg_earliest:step:dom}
				\If{$l_{new}$ dominates $l'$} \label{alg:alg_earliest:domcheck1} %
					\State remove $l'$ from $R_v$ and delete $\Pi'$
				\ElsIf {$l'$ dominates $l_{new}$} \label{alg:alg_earliest:domcheck2}
					\State delete $l_{new}$
					\State \textbf{goto} line \ref{alg:alg_earliest:step:break1}
				\ElsIf {$l'$ is equivalent to $l_{new}$} 
					\State set reference $\Pi\leftarrow \Pi'$
					\State insert $p_{new}$ into $\Pi'$ 
					\If{$t_e+\lambda_e<a'$}
						\State replace $l'$ in $R_v$ by $l_{new}$
					\EndIf
		
					\State \textbf{goto} line \ref{alg:alg_earliest:step:break1}				
				\EndIf
			\EndFor
			\State $\Pi\leftarrow \Pi_{l_{new}}$
			\State insert $l_{new}$ into $R_v$ and $u$ into $\Pi_{l_{new}}$
		\EndFor
		\Statex \textit{Phase 2}		
			\For{each label $l'=(0,a,c,p,z,r,\Pi_{l'}) \in R_z$}
				\For {each label $l\in \Pi_{l'}$ with minimal arrival time}
					\State initialize empty path $P$
					\State call \textsc{OutputPaths}($l$, $P$)\label{alg:alg_earliest:step:calloutput}
				\EndFor
			\EndFor
	\end{algorithmic}		
	\caption{for \probaenum}
	\label{alg:alg_earliest}
\end{algorithm}
\begin{theorem}\label{theorem:correct_alg1}
	Algorithm \ref{alg:alg_earliest} outputs exactly all efficient $(s,z)$-paths w.r.t. arrival time and costs.
\end{theorem}
\begin{proof}
	\Cref{lemma:effpathlables_alg2} implies that for each efficient path $P_{s,z}$ there is a corresponding representative label in $R_z$ after Phase 1 is finished. 
	For every representative $l'=(0,a,c,p,z,r,\Pi)$ in $R_z$, it holds by construction that all labels in $\Pi_{l'}$ have the same costs. Therefore, we only need to consider the nondominated labels with minimal arrival time $a_{min}$ over all labels in $\Pi_{l'}$.
	Hence, for each $l'\in R_z$ we call the output procedure for every label in $l\in \Pi_{l'}$ if $l$ has minimal arrival time $a_{min}$ in $\Pi_{l'}$. %
\end{proof}
\Cref{alg:alg_earliest} uses a linear number of labels.
\begin{lemma}\label{lemma:space_labels}
	The total number of labels generated and hold at the vertices in \Cref{alg:alg_earliest} is at most $m+1$.
\end{lemma}
\begin{proof}
	We need one initial label $l_{init}$ at the source vertex $s$.
	For each incoming edge $e=(u,v,t,\lambda,c)$ in the edge stream, in line \ref{alg:alg_earliest:step:break1} we choose the label $l$ with minimal costs and arrival time at most $t$.
	We only push $l$ and generate at most one new label $l_{new}$ at vertex $v$.
	Therefore, we generate at most $m+1$ labels in total.	
\end{proof}
\begin{lemma}\label{lemma:timephase1}
	Phase 1 of \Cref{alg:alg_earliest} has a time complexity of $\mathcal{O}(m^2)$.
\end{lemma}
\begin{proof}
	The outer loop iterates over $m$ edges. In each iteration we have to find the representative label $l\in R_u$ with minimum costs and arrival time $a\leq t_e$.
	This is possible in constant time, since we always keep the label with the earliest arrival time of each equivalence class as representative in $R_u$. 
	We only need to check if the arrival time of this label is earlier than the availability time of the current edge.
	Next we have to check the domination and equivalence between $l_{new}$ and each label $l'\in R_v$. Each of the cases takes constant time, and there are $\mathcal{O}(m)$ labels in $R_v$.
	Altogether, a time complexity of $\mathcal{O}(m^2)$ follows. 
\end{proof}
\Cref{alg:alg_earliest} lists all efficient paths in polynomial delay and uses only linear space.
\begin{theorem}\label{theorem:ocearl1}
	\probaenumsp{} $\in$ \psdelayp.
\end{theorem}
Using \Cref{alg:alg_earliest}, also the results of \Cref{theorem:algphase2poly} and \Cref{corollary:noneffpoly} %
can be adapted for the earliest arrival case.

\begin{theorem}\label{theorem:algphase2poly2}
	Given a weighted temporal graph $G=(V,E)$, $s,z\in V$ and
	\begin{enumerate}
		\item an $(s,z)$-path $P$, deciding if $P$ is efficient for \probaenum{}, or
		\item $c\in \mathbb{R}_{\geq 0}$ and $a\in \mathbb{N}$, deciding if there exists an $(s,z)$-path $P$ with $a(P)\leq a$ and $c(P)\leq c$
		is possible in polynomial time. 
	\end{enumerate}
\end{theorem}
\begin{corollary}\label{corollary:noneffpoly2}
	Given a temporal graph $G=(V,E)$ and $s,z\in V$, a 
	maximal set of efficient $(s,z)$-paths with pairwise different cost vectors for \probaenum{}
	can be found in time $\mathcal{O}(m^2)$. 
\end{corollary}

\section{Complexity of Counting Efficient Paths}\label{section:counting}
In this section, we discuss the complexity of counting efficient paths and show that 
the counting versions \probbcount{} and \probacount{} are both \spcomplete.
First, we show that already the unweighted earliest arrival temporal path problem is \spcomplete.
\\\\
\textsc{earliest Arrival $(s,z)$-Paths Counting Problem} (\tpeac)
\\\textbf{Input:} A temporal graph $\tg=(V,\tge)$ and $s,z\in V$.
\\\textbf{Output:} Number of fastest $(s,z)$-paths?
\begin{lemma}
	\tpeac{} is \emph{\spcomplete}.
\end{lemma}
\begin{proof}
	We provide a polynomial time Turing reduction from \textsc{ST-Path} to \tpeac. 
	\tpeac{} is the special case of \probacount{} where all edge costs are $0$.
    It is possible to decide if a path $P$ is efficient for \probas{} in polynomial time.
	This implies, that \tpeac{} is in \spc.
	Input of \textsc{ST-Path} is a static, directed graph $G=(V,E)$ and two vertices $s,t\in V$. The output is the number of simple paths from $s$ to $t$.
	Valiant~\cite{valiant1979complexityenum} showed that the problem is \spcomplete.
	Given a static, directed graph $G=(V,E)$ with $n\coloneqq|V|$ and two vertices $s,t\in V$, we iteratively construct $n-1$ temporal graphs $G_\tau=(V\cup\{z\}, E_\tau)$ with temporal edges 
	\[E_\tau\coloneqq \{ (u,v,i,1)\mid (u,v)\in E, 1\leq i\leq \tau\}\cup \{(t,z,n,1)\}\text{,}\] 
	for $\tau \in \{1,\ldots, n-1\}$. \Cref{fig:example_reduction_sp} shows an example of the construction. 
	Each temporal $(s,z)$-path in $G_\tau$ ends with edge $(t,z,n,1)$.
	This construction allows to determine all temporal $(s,z)$-paths of lengths $2\leq\ell\leq \tau+1$ in $G_\tau$, with $\tau+1$ being the maximal length of any path in $G_\tau$.
	For each $\tau \in \{1,\ldots, n-1\}$ let $y_\tau$ be the total number of earliest arrival $(s,z)$-paths in $G_\tau$.
	Note that $y_\tau$ equals the number of $(s,z)$-paths that use at most $\tau+1$ edges, because each edge traversal takes one time step.
	Each temporal $(s,z)$-path in $G_\tau$ corresponds to exactly one $(s,t)$-path in the static graph $G$ consisting of the same sequence of edges but the last edge. 
	Let $x_\tau$ be the number of temporal $(s,z)$-paths that use exactly $\tau+1$ edges. %
	Then it holds that the number of $(s,t)$-paths in $G$ equals $\sum_{\tau=1}^{n-1} x_\tau$.
	
	A path consisting of exactly $\tau+1$ edges in $G_\tau$ does not have waiting time at any vertex, beside possibly at vertex $t$. 
	But a path of length $\ell \leq \tau$ can wait for $\tau+1-\ell$ time steps at one or more vertices that it visits beside vertex $t$.
	In order to derive the value $x_\tau$ from $y_\tau$ we have to account for all paths of length $\ell \leq \tau$ that have such waiting times.
	Consider a path of length $\ell$ in $G_\tau$ with $\ell\leq\tau$. From the $\tau$ time steps, we choose $\ell$ time steps in which we do not wait and traverse an edge.
	Consequently, knowing that there exist $x_\ell$ paths of length $\ell$ in $G_\ell$, it follows that there are $x_\ell\cdot{\tau \choose \ell}$ paths of length $\ell$ with waiting times in $G_\tau$. %
	Based on these observations, the \Cref{alg:reduction} calculates $x_1,\ldots,x_{n-1}$, given an oracle for \tpeac{} in polynomial time.
\begin{algorithm}\mbox{\hfill}
	\\\textbf{Input:} Graph $G=(V,E)$, source $s\in V$ and target $t\in V$
	\\\textbf{Output:} Number of efficient $(s,t)$-paths
	\begin{algorithmic}[1]
	\For {$\tau=1,\ldots,n-1$}
		\State  Construct $G_\tau$
		\State  $y_\tau\coloneqq \#$Earliest arrival $(s,z)$-path in $G_\tau$ %
		\State  $x_\tau\coloneqq y_\tau - \displaystyle\sum_{i=1}^{\tau-1}{\tau \choose i}\cdot x_i$
	\EndFor
	\Return $\sum_{i=1}^{n-1}x_i$
	\end{algorithmic}		
	\caption{}
	\label{alg:reduction}
\end{algorithm}
	The algorithm iteratively determines the number $x_\tau$ of $(s,z)$-path with length $\tau+1$ for $\tau \in \{1,\ldots,n-1\}$. Each of these paths corresponds to exactly one $(s,t)$-path of length $\tau$ in $G$. Summing up all $x_i$ for $1\leq i\leq n-1$ leads exactly to the number of $(s,t)$-path in $G$.
\end{proof}
Now consider the \textsc{Fastest $(s,z)$-Path Counting Problem} (\tpfc{}). 
Both problems \tpeac{} and \tpfc{} are special case of the weighted versions \probacount{} and \probbcount{}. 
The parsimonious reduction of \Cref{lemma:levinred} implies that also \tpfc{} is {\spcomplete}.
Consequently, it follows \Cref{theorem:counting}.
\begin{theorem}\label{theorem:counting}
	\probbcount{} and \probacount{} are \emph{\spcomplete}. 
\end{theorem}
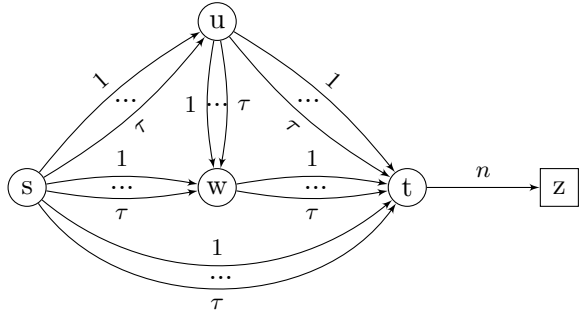
\begin{figure}
	\centering
	\begin{tikzpicture}
	\node[vertex] (u) at (-1.5,1.2) {u};
	\node[vertex] (w) at (-1.5,-1) {w};
	\node[vertex] (t) at (1,-1) {t};
	\node[vertex] (s) at (-4,-1) {s};
	\node[vertexb] (z) at (3.0,-1) {z};			
	\node at (-1.45,-2.2) {$...$};
	\node at (-2.75,-1) {$...$};
	\node at (-0.2,-1.0) {$...$};
	\node at (-1.5,0.1) {$...$};			
	\node at (-0.3,0.15) {$...$};	
	\node at (-2.7,0.15) {$...$};	
	
	\path[->] 
	(t)  edge[dedge]  node[midway,above,sloped] {$n$}   (z)
	(s)  edge[dedge, bend right=55]  node[midway,below,sloped] {$\tau$}   (t)
	(s)  edge[dedge, bend right=40]  node[midway,above,sloped] {$1$}   (t)	
	(u)  edge[dedge, bend left=10]  node[right] {$\tau$}   (w)
	(u)  edge[dedge, bend right=10]  node[left] {$1$}   (w)	
	(u)  edge[dedge, bend left=10]  node[midway,above,sloped] {$1$}   (t)	
	(u)  edge[dedge, bend right=10]  node[midway,below,sloped] {$\tau$}   (t)	
	(s)  edge[dedge, bend left=10]  node[midway,above,sloped] {$1$}   (w)	
	(s)  edge[dedge, bend right=10]  node[midway,below,sloped] {$\tau$}   (w)		
	(w)  edge[dedge, bend left=10]  node[midway,above,sloped] {$1$}   (t)
	(w)  edge[dedge, bend right=10]  node[midway,below,sloped] {$\tau$}   (t)	
	(s)  edge[dedge, bend left=10]  node[midway,above,sloped] {$1$}   (u)
	(s)  edge[dedge, bend right=10]  node[midway,below,sloped] {$\tau$}   (u);
	\end{tikzpicture}
	\caption{Example for the reduction from \textsc{ST-Path} to \tpeac{}. The instance of \textsc{ST-Path} consists of a graph $G=(V,E)$ with vertices $V=\{s,u,w,t\}$ and edges $E=\{(s,u),(s,w),(s,t),(u,w),(u,t),(w,t)\}$. In iteration $\tau$ of the reduction we use the oracle for \tpeac{} to find all $(s',z)$-path in the modified graph.} %
	\label{fig:example_reduction_sp}
\end{figure}
\section{Conclusion and Open Problems}\label{section:conclusion}
We discussed the bicriteria optimization problems 
\textsc{Min-Cost Earliest Arrival Paths Enumeration Problem} (\probaenum)
and
\textsc{Min-Cost Fastest Paths Enumeration} (\probbenum).
We have shown that enumerating exactly all efficient paths with low costs and early arrival time or short duration is possible in polynomial time delay and linear or polynomial space if the input graph has strictly positive edge costs. 
In case of nonnegative edge costs, it is possible to determine a maximal set of efficient paths with pairwise different cost vectors 
in $\mathcal{O}(m^2)$ time or $\mathcal{O}(\mathcal{S}\cdot m^2)$ time, respectively, where $\mathcal{S}$ is the number of distinct availability times of edges leaving the source vertex $s$.
We can find an efficient path for each nondominated point in polynomial time.

For the cases of zero-weighted or even negative edge weights, we cannot guarantee polynomial time delay for our algorithms to solve \probbenum{} or \probaenum{}. However, the proposed algorithms can be used to determine all efficient $(s,z)$-walks in polynomial time delay. Because each edge in a temporal graph can only be used for departure at a certain time, the number of different walks is finite and the algorithms terminate. So far, we are not aware of a way to ensure that only simple paths are enumerated without loosing the property that the delay between the output of two paths stays polynomially bounded.

Counting temporal paths is hard even in the unweighted single criterion case. We showed \spcompleteness{} for all discussed counting problems.
However, we are not aware of a reduction of a \spcomplete{} problem to the special case in which only strictly positive edge costs are allowed. %


\end{document}